\documentclass[journal,12pt,draftclsnofoot,onecolumn]{IEEEtran}
\usepackage{amsmath}
\usepackage{amssymb}
\usepackage{mathrsfs}
\usepackage{cite}
\usepackage{epsfig}
\usepackage{epsfig}
\usepackage{graphics}
\usepackage{hyperref}
\usepackage{epsfig}
\usepackage{setspace}
\newtheorem{thm}{Theorem}
\newtheorem{lem}{Lemma}

\newtheorem{defi}{Definition}

\begin{document}

\title{Secure Hybrid Digital-Analog Coding With Side Information at the Receiver}
\author{Ghadamali Bagherikaram, Konstantinos N. Plataniotis\\
The Edward S. Rogers Sr. Department of ECE, University of Toronto,\\
 10 King's College Road, Toronto, Ontario, Canada M5S 3G4\\
 Emails: \{gbagheri, kostas\}@comm.utoronto.ca
}
 \maketitle

\begin{abstract}
In this work, the problem of transmitting an i.i.d Gaussian source over an i.i.d Gaussian wiretap channel with an i.i.d Gaussian side information available at the intended receiver is considered. The intended receiver is assumed to have a certain minimum SNR and the eavesdropper is assumed to have a strictly lower SNR, compared to the intended receiver.  The objective is to minimize the distortion of source reconstruction at the intended receiver. In this work, it is shown that the source-channel separation coding scheme is optimum in the sense of achieving minimum distortion. Two hybrid digital-analog Wyner-Ziv coding schemes are then proposed which achieve the minimum distortion. These secure joint source-channel coding schemes are based on the Wyner-Ziv coding scheme and wiretap channel coding scheme when the analog source is not explicitly quantized. The proposed secure hybrid digital-analog schemes are analyzed under the main channel SNR mismatch. It is proven that the proposed schemes can give a graceful degradation of distortion with SNR under SNR mismatch, i.e., when the actual SNR is larger than the designed SNR.
\end{abstract}

\section{Introduction}

The notion of information
theoretic secrecy in communication systems was first introduced in \cite{1}.
The information theoretic secrecy requires that the received signal
by an eavesdropper not provide any information about the transmitted messages.
Following the pioneering works of \cite{2} and \cite{3} which have studied the wiretap
channel, many extensions of the wiretap channel model have been considered from a
perfect secrecy point of view (see e.g., \cite{4,5,6,7,8}). Particularly, in \cite{9,10}, the Gaussian wiretap channel of \cite{11} is extended to the Gaussian wiretap channel with side information available at the transmitter.

All extensions of the wiretap channel model have considered communicating a \emph{discrete} source with perfect secrecy constraint. In many applications, however, a bandlimited \emph{analog} source needs to be transmitted on a bandlimited Gaussian wiretap channel with side information available at the receiver. In many situations, the exact signal-to-noise ratio (SNR) of the main channel may not be known at the transmitter. Usually, a range of the main channel SNR is known but the true SNR value is unknown. Given a range of main channel SNR, such that the eavesdropper's signal is degraded with respect to the  legitimate receiver's signal, it is desirable to design a single transmitter which has a robust performance for all ranges of SNRs. A common method of designing such a system is based on Shannon's source-channel separation coding: Quantize the analog source and then transmit the resulting discrete source by the digital secret wiretap channel coding scheme. The main advantage of a digital system is that it is more reliable and cost efficient.

The inherent problem of digital systems is that they suffer from a severe form of "threshold effect" \cite{12,13}. This effect may be briefly described as follows: The system achieves a certain performance at a certain designed SNR. When the  SNR is increased, however, the system performance does not improve and it degrades drastically when the true SNR falls below the designed SNR. The severity of the threshold effect in digital systems is related to Shannon's source-channel separation principle \cite{14}. Recent works on non-secure communication systems, however, have proven that joint source-channel coding schemes can not only outperform the digital systems for a fixed complexity and delay, they are also more robust against the SNR variations. \cite{15,16,17,18,19,19_1}.

Some other schemes that exploit the advantage of analog systems are studied in \cite{19_2,19_3,19_4,19_5} and \cite{19_6}. These works are based on the so-called direct source-channel mapping technique which may briefly be explained as follows: the output of a source scalar/vector quantizer is mapped directly to a channel symbol using analog (or nearly analog) modulation, i.e., amplitude modulation or M-ary quadrature amplitude modulation (QAM). The direct source-channel codes have graceful degradation performance at low SNRs. In \cite{19_6}, a robust image coding system is presented which combines subband coding and QAM. This system allows various compression levels based on block-wise classification. An improved image coding system is then proposed in \cite{19_3} which utilizes both bandwidth compression and bandwidth expansion mappings, where the bandwidth expansion mapping employs a quantization error.

In \cite{17}, several hybrid digital-analog joint source-channel coding scheme are proposed for transmitting a Gaussian source over a (non-secure) Gaussian channel (without side information). The main idea in \cite{17} for increasing robustness is to reduce the number of quantization intervals, and thereby increase the distance between the decision lines of the quantization levels. This will increase the distortion, but in order to compensate the coarser representation, the quantization error is sent as an analog symbol using a linear coder (see also \cite{20}). In \cite{21}, different coding schemes are analyzed for transmitting a Gaussian source over a Gaussian wiretap channel (without side information). For a fixed information leakage rate to the eavesdropper, \cite{21} has shown that superimposing the secure digital signal with the analog (quantization error) part has better performance compared to the separation based scheme and the uncoded scheme. In \cite{22}, the problem of transmitting a Gaussian source over a (non-secure) Gaussian channel with side information (either at transmitter or at receiver) is studied. \cite{22} has introduced several hybrid digital-analog forms of the Costa and Wyner-Ziv coding (\cite{23}) schemes. In \cite{22}, the results of \cite{24} are extended to the case in which the transmitter or receiver has side information, and have shown that there are infinitely many schemes for achieving the optimal distortion. In the work of \cite{25}, we considered the problem of transmitting a Gaussian source over a \emph{secure} Gaussian channel with side information available at the transmitter and proposed different secure joint source channel coding schemes based on the secret dirty paper coding and wiretap channel coding scheme.

In this paper, we consider the problem of transmitting an i.i.d Gaussian source over an i.i.d Gaussian wiretap channel with side information available at the intended receiver. We assume that the intended receiver has a certain minimum SNR and the eavesdropper has a strictly lower SNR compared to the intended receiver. We are interested in minimizing the distortion of source reconstruction at the intended receiver.  We show that, here, like the Gaussian wiretap channel without side information, Shannon's source-channel separation coding scheme is optimum in the sense of achieving the minimum distortion. We then propose two hybrid digital-analog secure joint source-channel coding schemes which achieve the minimum distortion. Our coding schemes are based on the Wyner Ziv coding scheme and wiretap channel coding scheme when the analog source is not explicitly quantized. We will illustrate that these schemes achieve the optimum distortion. We analyze our secure hybrid digital-analog schemes under the main channel SNR mismatch. We will show that our proposed schemes can give a graceful degradation of distortion with SNR under SNR mismatch, i.e., when the actual SNR is larger than the designed SNR.
\section{Preliminaries And Related Works}

\subsection{Notation}

In this paper, random variables are denoted by capital letters (e.g.
$X$) and their realizations are denoted by corresponding lower case
letters (e.g. $x$). The finite alphabet of a random variable is
denoted by a script letter (e.g. $\mathcal{X}$) and its probability
distribution is denoted by $P(x)$. Similarly, the function $P(x,y)$ represents the joint probability distribution function of the random variables $X$ and $Y$. The vectors will be written as
$x^{n}=(x_{1},x_{2},...,x_{n})$, where subscripted letters denote
the components and superscripted letters denote the vector. The notation
$x_{i}^{j}$ denotes the vector $(x_{i},x_{i+1},...,x_{j})$ for $j\geq i$. A Gaussian
Random variable $X$ with a mean of $\mu$ and variance of $\sigma^{2}$ is denoted by $X\sim \mathcal{N}(\mu,\sigma^2)$. The function $E[.]$ represents a statistical expectation. The function $I(X;Y)$ represents mutual information between random variables $X$ and $Y$ and $A^{*(n)}_{\epsilon}$ denotes the set of strongly jointly typical sequences.

\subsection{System Model And Problem Statement}

\underline{Source Model}: Consider a memoryless Gaussian source of $\{V_{i}\}_{i=1}^{\infty}$ with zero mean and variance $\sigma^{2}_{v}$. Thus, $V_{i}\sim \mathcal{N}(0,\sigma^{2}_{v})$ and we assume that the sequence $\{V_{i}\}$ is independent and identically distributed (i.i.d). We assume that the source is obtained from uniform sampling of a continuous-time Gaussian process with bandwidth $W_{s}(Hz)$. Furthermore, we assume that the sampling rate is $2W_{s}$ samples per second.

\underline{Channel Model}: The source is transmitted over an Additive White Gaussian Noise (AWGN) wiretap channel when the intended receiver has some side information about the source. The channel is therefore modeled as follows:
\begin{IEEEeqnarray}{rl}
Y_{i}&=X_{i}+W_{i},\\ \nonumber
Z_{i}&=X_{i}+W_{i}^{'},\\ \nonumber
V_{i}&=V_{i}^{'}+T_{i},
\end{IEEEeqnarray}
where $X_{i}$, $Y_{i}$, and $Z_{i}$ are the channel input, the received signal by the intended receiver and the received signal by the eavesdropper, respectively. We assume that $E[X_{i}^{2}]\leq P$, and $W_{i}\sim\mathcal{N}(0,N_{1})$, $W_{i}^{'}\sim\mathcal{N}(0,N_{2})$, where $N_{2}>N_{1}$. Furthermore, assume that $T_{i}$'s are a sequence of real i.i.d Gaussian random variables with zero mean and variance $\sigma_{t}^{2}$, i.e. $T_{i}\sim \mathcal{N}(0,\sigma_{t}^{2})$. Here $V_{i}^{'}$ and $T_{i}$ are mutually independent Gaussian random variables. As the source, side information and the channel are i.i.d over the time, we will omit the index $i$ throughout the rest of the paper. The channel is derived from a continuous-time AWGN wiretap channel with bandwidth $W_{c}(Hz)$. The equivalent discrete-time channel is used at a rate of $2W_{c}$ channel uses per second. The block diagram of the system is depicted in Fig.1.

\begin{figure}
\leftline{\includegraphics[scale=.22]{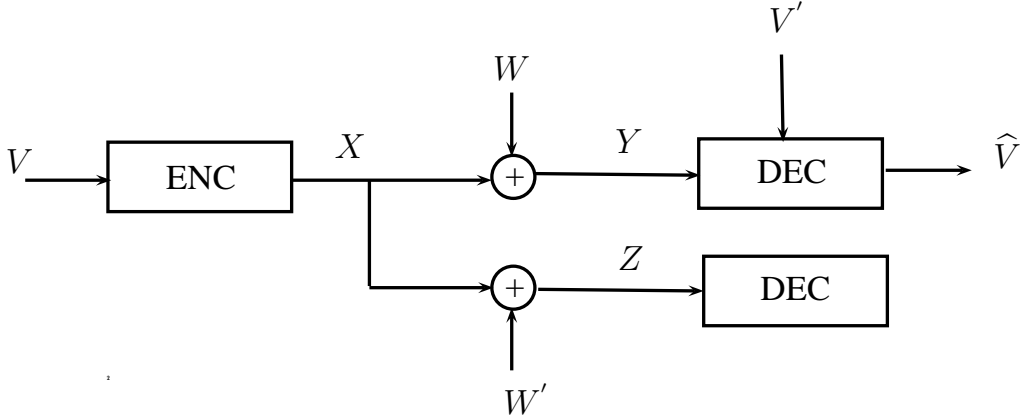}} \caption{Block diagram of the secure joint source-channel coding problem with side information known at the receiver.}
\end{figure}

\underline{Coding Scheme}: The source samples are grouped into blocks of size $m$
\begin{equation}\nonumber
V^{m}=(V_{1},V_{2},...,V_{m}),
\end{equation}
and the encoder is a mapping $f_{m}:\mathbb{R}^{m}\rightarrow \mathbb{R}^{n}$ which satisfies the power constraint $E[\|f_{m}(V^{m})\|^{2}]\leq nP$. Let us define the parameter $\rho=n/m=W_{c}/W_{s}$. In this paper we assume that $\rho=1$, i.e., $m=n$. The received signals by the intended receiver and the eavesdropper are given by
\begin{IEEEeqnarray}{rl}
Y^{n}=X^{n}+W^{n},\\ \nonumber
Z^{n}=X^{n}+W^{'n},
\end{IEEEeqnarray}
where $X^{n}=f_{n}(V^{n})$, $W^{n}\sim \mathcal{N}(0,N_{1}\mathbf{I_{n}})$,  $W^{'n}\sim \mathcal{N}(0,N_{2}\mathbf{I_{n}})$, and $\mathbf{I_{n}}$ is the $n\times n$ identity matrix. The decoder at the intended receiver is a mapping $g_{n}:\mathbb{R}^{n}\rightarrow \mathbb{R}^{n}$. The average squared-error distortion of the coding scheme at the intended receiver is given by
\begin{equation}
\bar{D}_{n}(f_{n},g_{n},N_{1},N_{2})=\frac{1}{n}E[\|V^{n}-\hat{V}^{n}\|^{2}],
\end{equation}
where $\hat{V}^{n}=g_{n}(Y^{n})$. For the purpose of analysis, we will consider sequences of codes $(f_{n},g_{n})$, where $n$ is increasing. The asymptotic performance of the code is given by
\begin{equation}
\bar{D}(N_{1},N_{2})=\lim_{n\rightarrow \infty}\bar{D}_{n}(f_{n},g_{n},N_{1},N_{2}).
\end{equation}
Note that the above $\bar{D}$ is also a function of $\sigma^{2}_{v}>0$, $P>0$, $\sigma_{t}^{2}>0$, and $\rho>0$, but we assume that these parameters are known and fixed, and therefore express $\bar{D}$ as a function of $(N_{1},N_{2})$. In subsequent sections, we refer to $\bar{D}$ as mean-squared distortion and omit the bar superscript and denote it by $D$, i.e., $D=\bar{D}$.

\underline{Secrecy Requirements}: The eavesdropper obtains $Z^{n}$. The secrecy of the system is measured by the information leaked to the eavesdropper and is expressed as follows
\begin{IEEEeqnarray}{rl}\label{Ie}
I_{\epsilon}=\frac{1}{n}I(V^{n};Z^{n}).
\end{IEEEeqnarray}
Note that $I_{\epsilon}=0$ corresponds to the perfect secrecy condition and implies that the eavesdropper obtains no information about the source. In this paper we consider the situation in which the leakage information $I_{\epsilon}$ is known and is a fixed constant.

\underline{Distortion Exponent}: In practical scenarios, the transmitter usually does not have an exact knowledge of the actual noise power at the intended receiver and the eavesdropper. The transmitter, however, knows an upperbound for the noise power of the intended receiver and also knows a lowerbound for the noise power of the eavesdropper. Let us denote the upperbound noise power of the intended receiver and the lowerbound noise power of the eavesdropper by  $N_{1}$ and $N_{2}$, respectively. Therefore, the transmitter designs its coding scheme based on the noise power of  $N_{1}$ and $N_{2}$ such that $N_{1}\geq N_{1a}$ and $N_{2}\leq N_{2a}$, where $N_{1a}$ and $N_{2a}$ are the actual noise variances corresponding to the actual $SNR_{1a}=\frac{P}{N_{1a}}$ and $SNR_{2a}=\frac{P}{N_{2a}}$, respectively. The eavesdropper channel is still a degraded version of the main channel and is assumed to have the lowest $SNR_{2a}<SNR_{2}<SNR_{1}<SNR_{1a}$, where $SNR_{2}=\frac{P}{N_{2}}$ and $SNR_{1}=\frac{P}{N_{1}}$. The intended receiver is assumed to have a perfect estimate of $SNR_{1a}$, but the transmitter communicates at a lower designed $SNR_{1}$. In this scenario, we expect a graceful degradation of distortion $D(SNR_{1a}, SNR_{2})$ with $SNR_{1a}$ compared with $D(SNR_{1}, SNR_{2})$ when the actual $SNR_{1a}>SNR_{1}$. Let us define the distortion exponent as follows:

\begin{defi}
For a fixed $SNR_{2}$, the distortion exponent of $D(SNR_{1a},SNR_{2})$ is given by
\begin{equation}\label{eqde}
\zeta\stackrel{\triangle}{=}-\lim_{SNR_{1a}\rightarrow \infty}\frac{\log D(SNR_{1a},SNR_{2})}{\log SNR_{1a}}.
\end{equation}
\end{defi}

The highest possible distortion exponent is $\rho$ and therefore, $0\leq \zeta\leq \rho$. The distortion exponent can be used as a criterion for the robustness of a coding scheme. A high distortion exponent means that the coding scheme is more robust against the case of $SNR$ mismatch where we design the scheme to be optimal for a channel noise variance of $N_{1}$, but the actual noise variance is $N_{1a}$. In this paper, we propose two robust coding schemes which achieve the optimum mean-squared distortion and analyze them for $SNR$ mismatch. Before introducing our proposed schemes, we need to review some related works in this area.

\subsection{Related Works}

\subsubsection{Digital Wiretap Channel}

In a digital wiretap channel (without any interference $S$), a digital message $M\in\{1,2,...,nC_{s}\}$ is transmitted to the intended receiver while the eavesdropper is kept ignorant. Wyner in \cite{2} characterized the secrecy capacity of this channel when the eavesdropper's channel is degraded with respect to the main channel. Csiszar et. al. in \cite{3} considered the general wiretap channel and established its secrecy capacity. Let us assume $X$, $Y$ and $Z$ to be the channel input, intended receiver's signal and eavesdropper's signal, respectively. The secrecy capacity of a wiretap channel is given by
\begin{equation}
C_{s}=\max_{P(u,x)}2W_{c}\left[I(U;Y)-I(U;Z)\right],
\end{equation}
where $U\rightarrow X\rightarrow YZ$ forms a Markov chain. When the channels are AWGN, \cite{11} has shown that the secrecy capacity is given by
\begin{equation}
C_{s}=W_{c}\left[\log(1+\frac{P}{N_{1}})-\log(1+\frac{P}{N_{2}})\right].
\end{equation}

Here, we briefly explain the coding scheme. We Generate $2^{nI(U;Y)}$ Gaussian codewords $U^{n}$ and throw them uniformly at random into $2^{nC_{s}}$ bins. Each bin thus contains $2^{nI(U;Z)}$ codeword $U^{n}$. To encode the message $M\in\{1,2,...,2^{C_{s}}\}$  randomly choose a $U^{n}$ from the bin which is indicated by $M$ and send it. The intended receiver seeks for a $U^{n}$ which is jointly typical with $Y^{n}$ and declares the bin index as the transmitted message. The probability of error asymptotically tends to be zero, i.e., $\lim_{n\rightarrow \infty}P_{e}(\hat{M}\neq M)\rightarrow 0$. The information leakage is $\lim_{n\rightarrow \infty}\frac{1}{n}I(M;Z^{n})=0$. When $I_{\epsilon}$ is not zero but is a fixed known constant, we have the following Lemma:

\begin{lem}
In a digital wiretap channel, the information leakage rate to the eavesdropper is $I_{\epsilon}$ for the rate of
\begin{IEEEeqnarray}{rl}\label{RIe}
R_{I_{\epsilon}}&=C_{s}+I_{\epsilon}\\
\nonumber &=W_{c}\left[\log(1+\frac{P}{N_{1}})-\log(1+\frac{P}{N_{2}})\right]+I_{\epsilon}.
\end{IEEEeqnarray}
\end{lem}
\begin{proof}
Let $M$ be the secret message. The transmission rate is then given by
\begin{IEEEeqnarray}{rl}\label{seq1}
R_{I_{\epsilon}}=\frac{H(M)}{n}.
\end{IEEEeqnarray}
According to the security analysis of \cite{13}, by applying Theorem 1 of \cite{13} to our problem setup we obtain the following inequality,
\begin{IEEEeqnarray}{rl}\label{seq2}
\frac{H(M|Z^{n})}{n}\leq W_{c}\left[\log(1+\frac{P}{N_{1}})-\log(1+\frac{P}{N_{2}})\right].
\end{IEEEeqnarray}
By combining (\ref{seq1}) and (\ref{seq2}), we have
\begin{IEEEeqnarray}{rl}
R_{I_{\epsilon}}\frac{H(M|Z^{n})}{H(M)}\leq W_{c}\left[\log(1+\frac{P}{N_{1}})-\log(1+\frac{P}{N_{2}})\right].
\end{IEEEeqnarray}
The left side of the above inequality can be simplified as follows:
\begin{IEEEeqnarray}{rl}
R_{I_{\epsilon}}\frac{H(M|Z^{n})}{H(M)}&=R_{I_{\epsilon}}\frac{H(M|Z^{n})-H(M)+H(M)}{H(M)}\\ \nonumber &\stackrel{(a)}{=}R_{I_{\epsilon}}\frac{H(M)-nI_{\epsilon}}{H(M)}\\ \nonumber &\stackrel{(b)}{=}R_{I_{\epsilon}}-I_{\epsilon},
\end{IEEEeqnarray}
where $(a)$ follows from the definition of $I_{\epsilon}$ in (\ref{Ie}), and $(b)$ follows from (\ref{seq1}). Thus, we can bound $R_{I_{\epsilon}}$ as
\begin{IEEEeqnarray}{rl}
R_{I_{\epsilon}}\leq W_{c}\left[\log(1+\frac{P}{N_{1}})-\log(1+\frac{P}{N_{2}})\right]+I_{\epsilon}.
\end{IEEEeqnarray}
Therefore, the maximum rate for $R_{I_{\epsilon}}$ can be obtained by choosing
\begin{IEEEeqnarray}{rl}
R_{I_{\epsilon}}=W_{c}\left[\log(1+\frac{P}{N_{1}})-\log(1+\frac{P}{N_{2}})\right]+I_{\epsilon}.
\end{IEEEeqnarray}

\end{proof}

 \subsubsection{Digital Wyner-Ziv Coding}

Consider a (non-secure) source coding problem with side information known at the receiver. The rate distortion function with side information $R_{V^{'}} (D)$
is defined as the minimum rate required to achieve distortion $D$ if the
side information $V^{'}$ is available to the decoder. Precisely, $R_{V^{'}} (D)$ is the
infimum of rates $R$ such that there exist maps in $f_{n}: \mathcal{V}^{n}\rightarrow \{1,... , 2^{nR}\},$ and $ g_{n} : \mathcal{V}^{'n} \times \{1, . . . , 2^{nR}\}\rightarrow \mathcal{V}^{n}$ such that
\begin{IEEEeqnarray}{rl}
\lim_{n\rightarrow\infty}\sup Ed(V^{n}, g_{n}(V^{'n},f_{n}(V^{n})))\leq D.
\end{IEEEeqnarray}
Wyner-Ziv coding scheme achieved the entire curve $R_{V^{'}}(D)$ as in the
following theorem.

\begin{thm}\cite{2} $\left[\hbox{Rate distortion with side information}\right]$
Let $(V,V^{'})$ be drawn i.i.d. according to the joint distribution $p(v,v^{'})$ and let $d(v^{n},\hat{v}^{n})=\frac{1}{n}\sum_{i=1}^{n} d(v_{i} ,\hat{v}_{i})$ be given. The rate distortion function with side information is given by
\begin{IEEEeqnarray}{rl}
R_{V^{'}} (D) =\min_{p(u|v)}\min_{g}\left\{I(U;V) - I (U;V^{'})\right\}
\end{IEEEeqnarray}
where the minimization is over all functions $g : \mathcal{V}^{'} \times \mathcal{U} \rightarrow \hat{\mathcal{V}}$ and conditional
probability mass functions $p(u|v), |\mathcal{U}|\leq |\mathcal{V}| + 1$, such that
\begin{IEEEeqnarray}{rl}
\sum_{v}\sum_{u}\sum_{v^{'}}p(v, v^{'})p(u|v)d(v, g (v^{'},u))\leq D.
\end{IEEEeqnarray}
\end{thm}

The function $g$ in the theorem corresponds to the decoding map that
maps the encoded version of the $V$ symbols and the side information $V^{'}$ to
the output alphabet. We minimize over all conditional distributions on $U$
and functions $g$ such that the expected distortion for the joint distribution
is less than $D$. Here, we briefly explain the Wyner-Ziv achievability scheme: Fix $p(u|v)$ and the function $g (u,v^{'})$. Calculate $p(u) =\sum_{v} p(v)p(u|v)$.

\emph{Codebook Generation}: Let $R_{1} = I (V;U) + \epsilon$. Generate $2^{nR_{1}}$ i.i.d.
codewords $U^{n}(s) \sim \prod_{i=1}^{n} p(u_{i})$, and index them by $s\in\{1,2,...,2^{nR_{1}}\}$. Let $R_{2}=I(U;V)-I(U;V^{'})+5\epsilon$. Randomly assign the indices $s\in\{1,2,..., 2^{nR_{1}}\}$ to one of $2^{nR_{2}}$ bins using a uniform distribution over the bins. Let $B(i)$ denote the indices assigned to bin $i$. There are approximately $2^{n(R1-R2)}$ indices in each bin.

\emph{Encoding}: Given a source sequence $V^{n}$, the encoder looks for a codeword
$U^{n}(s)$ such that $(V^{n},U^{n}(s))\in A^{*(n)}_{\epsilon}$ . If there is no such $U^{n}$, the encoder sets $s = 1$. If there is more than one such $s$, the encoder uses the lowest $s$. The encoder sends the index of the bin in which $s$ belongs.

\emph{Decoding}: The decoder looks for a $U^{n}(s)$ such that $s \in B(i)$ and
$(U^{n}(s), V^{'n})\in A^{*(n)}_{\epsilon}$ . If it finds a unique $s$, it then calculates $\hat{V}^{n}$, where $\hat{V}_{i}= f (U_{i}, V^{'}_{i})$. If it does not find any such $s$ or more than one such $s$, it sets $\hat{V}^{n} = \hat{v}^{n}$, where $\hat{v}^{n}$ is an arbitrary sequence in $\hat{\mathcal{V}}^{n}$. It does not
matter which default sequence is used; it is shown that the probability
of this event is small.

\section{Separation Based Scheme With Wyner-Ziv Coding and the Pure Analog Scheme}

In this section, we first analyzed the separation based scheme in which the source is initially quantized and then transmitted using an optimum channel coding scheme. Next, we evaluate an uncoded analog scheme in which an optimum scaling version of the the source signal is transmitted.

\subsection{secure digital Wyner-Ziv coding scheme}
In this section we consider the problem described in Fig. 1. We use a separation scheme with an optimum Wyner-Ziv code followed by an optimum channel code. We refer to this scheme as the "secure digital Wyner-Ziv coding scheme". We show that this scheme achieves the optimum possible distortion.

If we suppose that the side information $V^{'}$ is also available at the encode, the transmitter therefore requires sending the remaining information $T$. In \cite{6} the rate distortion problem for Shannon cipher system is considered. It can be seen from Theorem 1 in \cite{6} by setting $R_{k}=0$, the Shannon cipher system reduces to the wiretap channel setup and the optimum distortion can be achieved by separate source coding followed by digital wiretap channel coding scheme. Therefore, for a fixed leaked information $I_{\epsilon}$ we first need to quantize the source $T^{n}$ to $T^{n}_{q}$ at a rate $C_{s}=\frac{1}{2}\log\left(1+\frac{P}{N_{1}}\right)-\frac{1}{2}\log\left(1+\frac{P}{N_{2}}\right)$ and then transmit $V^{n}_{q}$. Thus, we can achieve the optimum distortion of
\begin{IEEEeqnarray}{rl}
D^{*}&=\sigma_{t}^{2}2^{-2\left(C_{s}+I_{\epsilon}\right)}\\ \nonumber
&=\frac{\sigma^{2}_{t}}{(1+\frac{P}{N_{1}})/(1+\frac{P}{N_{2}})2^{2I_{\epsilon}}}.
\end{IEEEeqnarray}

In the problem of Fig.1., assume that a genie informs the transmitter about the side information $V^{'}$. Therefore, the above distortion is an upperbound for the problem of Fig.1, i.e., $D^{*}$ is the minimum possible value for distortion $D=E\left[\|V-\hat{V}\|^{2}\right]$. We show that the same distortion can be achieved in the setup of Fig.1 by using the secure digital Wyner-Ziv coding scheme. Fig.2. $a)$ shows the separation based scheme. The following theorem illustrates this result.
\begin{figure}
\leftline{\includegraphics[scale=.2]{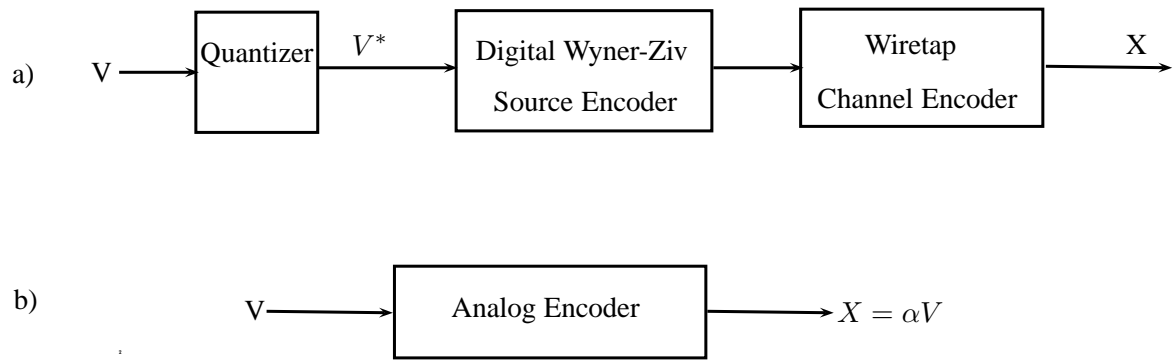}} \caption{a) Block Diagram of the Encoder of the Separation Based Scheme. b) Block Diagram of the Analog Encoder.}
\end{figure}

\begin{thm}
In the problem of secure transmitting a discrete-time Gaussian analog source over a Gaussian wiretap channel with side information known at the intended decoder, a separation based scheme can achieve the optimum distortion of $D^{*}$.
\end{thm}

\begin{proof}
The achievability scheme is secure digital Wyner-Ziv coding scheme. Let $U$ be an auxiliary random variable given by
\begin{IEEEeqnarray}{rl}
U=\sqrt{\alpha}V+F,
\end{IEEEeqnarray}
where $F\sim \mathcal{N}(0,D)$. We generate an $n$-length i.i.d Gaussian codebook $\mathcal{U}$ with $2^{nI(U,V)}$ codewords, where each component of a codeword is Gaussian with zero mean and variance $\alpha\sigma^{2}_{v}+D$. We then evenly distribute them into $2^{nR_{V^{'}}(D)}$ bins, where $R_{V^{'}}(D)=I(U;V)-I(U;V^{'})$. Let $i(u^{n})$ be the index of the bin containing $u^{n}$. For each $v^{n}$, find a $u^{n}$ such that $(u^{n},v^{n})$ are strongly jointly typical, i.e., $(u^{n},v^{n})\in A^{*(n)}_{\epsilon}$. The index $i(u^{n})$ is the Wyner-Ziv source coding index. The transmitter then encoded the index $i(u^{n})$ using an optimal secure channel code of rate arbitrary close to $C_{s}$ and transmit it over the channel. The receiver decodes the index $i(u^{n})$ with high probability. Next for the decoded $i(u^{n})$ we look for an $u^{n}$ in the bin whose index is $i(u^{n})$ such that $(u^{n},v^{'n})\in A^{*(n)}_{\epsilon}$. We make an estimate for source $v^{n}$ from the decoded $u^{n}$ and $v^{'n}$ as follows:
\begin{IEEEeqnarray}{rl}
\hat{v}^{n}=\left(
              \begin{array}{cc}
                \lambda_{1} & \lambda_{2} \\
              \end{array}
            \right)
\left(
                                       \begin{array}{c}
                                         u^{n} \\
                                         v^{'n} \\
                                       \end{array}
                                     \right),
\end{IEEEeqnarray}
where we need to determine $\lambda_{1}$ and $\lambda_{2}$ such that the distortion $D=E\{|v^{n}-\hat{v}^{n}|^{2}\}$ is minimized. After some math, the optimum values for $\lambda_{1}$ and $\lambda_{2}$ are as follows:
\begin{IEEEeqnarray}{rl}
\lambda_{1}&=\sqrt{\alpha}\\ \nonumber
\lambda_{2}&=1-\alpha,
\end{IEEEeqnarray}
and the related distortion is given by
\begin{IEEEeqnarray}{rl}\label{e1}
D=(1-\alpha)\sigma^{2}_{t}.
\end{IEEEeqnarray}

To calculate the coefficient $\alpha$, note that
\begin{IEEEeqnarray}{rl}
R_{V^{'}}(D)&=I(U;V)-I(U;V^{'})\\ \nonumber
&=\frac{1}{2}\log\left(1+\frac{\alpha\sigma^{2}_{v}}{D}\right)-\frac{1}{2}\log\left(1+\frac{\alpha\sigma^{2}_{v^{'}}}{\alpha\sigma^{2}_{t}+D}\right)\\ \nonumber &=\frac{1}{2}\log\left(\frac{D+\alpha\sigma^{2}_{v}}{\alpha\sigma^{2}_{t}+\alpha\sigma^{2}_{v^{'}}+D}\frac{\alpha\sigma^{2}_{t}+D}{D}\right)\\ \nonumber &\stackrel{(a)}{=}\frac{1}{2}\log\left(\frac{\alpha\sigma^{2}_{t}+D}{D}\right),
\end{IEEEeqnarray}
where $(a)$ follows from the fact that $\sigma^{2}_{v}=\sigma^{2}_{v^{'}}+\sigma^{2}_{t}$.

According to the fact that  $R_{V^{'}}(D)$ must be equal to $C_{s}+I_{\epsilon}$ we have,

\begin{IEEEeqnarray}{rl}\label{e2}
\alpha=\frac{D}{\sigma^{2}_{t}}\left[2^{2I_{\epsilon}}\frac{P+N_{1}}{P+N_{2}}\frac{N_{2}}{N_{1}}-1\right]. \end{IEEEeqnarray}

From ($\ref{e1}$) and ($\ref{e2}$) the achieved distortion is given by
\begin{IEEEeqnarray}{rl}
D&=\frac{\sigma^{2}_{t}}{2^{2I_{\epsilon}}\frac{P+N_{1}}{P+N_{2}}\frac{N_{2}}{N_{1}}}\\ \nonumber
&=D^{*}
\end{IEEEeqnarray}
\end{proof}

\subsection{Uncoded Scheme}

In this section we analyze a simple uncoded scaling scheme $($Fig.2. $b))$ in which the transmitter signal is given by
\begin{IEEEeqnarray}{rl}
X=\alpha V,
\end{IEEEeqnarray}
where $\alpha\leq \sqrt{\frac{P}{\sigma^{2}_{v}}}$. The intended receiver's and the eavesdropper's signals are therefore given by
\begin{IEEEeqnarray}{rl}
Y&=\alpha V+ W,\\ \nonumber
Z&= \alpha V + W^{'}.
\end{IEEEeqnarray}
For a fixed leakage information, $I_{\epsilon}$ is given as follows:
\begin{IEEEeqnarray}{rl}
I_{\epsilon}=I(V;Z)=\frac{1}{2}\log\left(1+\frac{\alpha^{2}\sigma^{2}_{v}}{N_{2}}\right).
\end{IEEEeqnarray}
From the above equation we can see that the choice of $\alpha=\sqrt{\frac{P}{\sigma^{2}_{v}}}$ results in a reasonably high information leakage rate. We can, however, reduce the value of $\alpha$ to satisfy our fixed information leakage to the eavesdropper.  Therefore, $\alpha$ is given by
\begin{IEEEeqnarray}{rl}\label{alfa}
\alpha=\sqrt{\frac{N_{2}(2^{2I_{\epsilon}}-1)}{\sigma^{2}_{v}}}.
\end{IEEEeqnarray}
The intended receiver makes an estimation of $V$ based on the received signal $Y$ and the side information $V^{'}$. The estimated signal is therefore as follows:
\begin{IEEEeqnarray}{rl}
\hat{V}&=\lambda_{1}Y+\lambda_{2}V^{'}\\ \nonumber
&= \alpha \lambda_{1} V+ \lambda_{1}W +\lambda_{2}V^{'}.
\end{IEEEeqnarray}
The corresponding distortion $D_{u}$ is given by
\begin{IEEEeqnarray}{rl}\label{dist}
D_{u}&=E\left[\left|\hat{V}-V\right|^{2}\right]\\ \nonumber
&=E\left[\left(\alpha \lambda_{1}-1\right)V+\lambda_{1}W + \lambda_{2}V^{'}\right]\\ \nonumber
&=\left(\alpha\lambda_{1}-1\right)^{2}\sigma^{2}_{v}+\lambda_{1}^{2}N_{1}+\lambda_{2}^{2}\sigma^{2}_{v^{'}}+2\lambda_{2}\left(\alpha\lambda_{1}-1\right)\sigma^{2}_{v^{'}}.
\end{IEEEeqnarray}
After some algebra, the optimum values for $\lambda_{1}$ and $\lambda_{2}$ which minimize distortion $D_{u}$ are are given as follows:
\begin{IEEEeqnarray}{rl}\label{lamdas}
\lambda_{1}&=\frac{\alpha\sigma_{t}^{2}}{N_{1}+\alpha^{2}\sigma_{t}^{2}}\\ \nonumber
\lambda_{2}&= \frac{N_{1}}{N_{1}+\alpha^{2}\sigma_{t}^{2}}.
\end{IEEEeqnarray}
By substituting (\ref{alfa}) and (\ref{lamdas}) into (\ref{dist}), the minimum distortion in this uncoded scheme is given by
\begin{IEEEeqnarray}{rl}\label{uncd}
D_{u}=\frac{\sigma^{2}_{t}}{1+\frac{\sigma_{t}^{2}}{\sigma_{v}^{2}}\frac{N_{2}}{N_{1}}\left(2^{2I_{\epsilon}}-1\right)}.
\end{IEEEeqnarray}

\section{Secure Joint Source-Channel Coding Schemes with Side Information at the Receiver}

In this section, we propose two different secure joint source-channel coding schemes which achieve the minimum distortion of $D^{*}$. The first proposed scheme does not involve quantizing the source explicitly while the second proposed scheme is a superimposed digital and the first hybrid scheme.

\subsection{Secure Hybrid Digital-Analog Wyner-Ziv Coding Scheme}

In this scheme, the encoder does not quantize the source explicitly, however it treats the source signal as a digital message. We generate the auxiliary random variable $U$ as follows:
\begin{IEEEeqnarray}{rl}
U=X+kV,
\end{IEEEeqnarray}
where $k$ is defined as $k^{2}=\frac{1}{\sigma^{2}_{t}}\left[\frac{PN_{2}}{P+N_{2}}2^{2I_{\epsilon}}-\frac{PN_{1}}{P+N_{1}}\right]$ and $X\sim\mathcal{N}(0,P)$.

\emph{Codebook Generation:} We generate a random i.i.d codebook $\mathcal{U}$ with $2^{nI(U;V)}$ sequences, where the components of the codewords are zero mean Gaussian random variables with variance $P+k^{2}\sigma_{v}^{2}$. We then distribute the generated codewords into $2^{nR}$ bins. This codebook is shared between the encoder and the intended receiver's decoder.

\emph{Encoding:} For a given $v^{n}$ the transmitter finds  $u^{n}$'s such that $(u^{n},v^{n})$ are strongly jointly typical, i.e., $(u^{n},v^{n})\in A^{*(n)}_{\epsilon}$. The transmitter randomly chooses one of $u^{n}$ and then sends $x^{n}=u^{n}-kv^{n}$. This is possible with arbitrary high probability if $R>I(U;V)-I(U;V^{'})$

\emph{Decoding:} The intended receiver's signal is $y^{n}=x^{n}+w^{n}$. The intended decoder finds a $u^{n}$ such that $(v^{'n},y^{n},u^{n})\in A^{*(n)}_{\epsilon}$. A unique such $u^{n}$ can be found with high probability if $R<I(U;V^{'},Y)-I(U;Z)$. We next show that we can choose $R$ to satisfy $I(U;V)-I(U;V^{'})<R<I(U;V^{'},Y)-I(U;Z)$. Equivalently, we show that with $k^{2}=\frac{P^{2}(N_{2}-N_{1})}{\sigma^{2}_{t}(P+N_{1})(P+N_{2})}$, we have $I(U;V)-I(U;V^{'})<I(U;V^{'},Y)-I(U;Z)$. Note that after some manipulation
\begin{IEEEeqnarray}{rl}\nonumber
I(U;V)-I(U;V^{'})&=\frac{1}{2}\log\left(\frac{P+k^{2}\sigma_{v}^{2}}{P}\right)-\frac{1}{2}\log\left(\frac{P+k^{2}\sigma_{v}^{2}}{P+k^{2}\sigma_{t}^{2}}\right)\\ \nonumber &=\frac{1}{2}\log\left(\frac{P+k^{2}\sigma_{t}^{2}}{P}\right),
\end{IEEEeqnarray}
and
\begin{IEEEeqnarray}{rl}\nonumber
I(U;V^{'},Y)-I(U;Z)=\frac{1}{2}\log\left(1+P\frac{\sigma_{v}^{2}}{\sigma_{t}^{2}}\frac{N_{2}-N_{1}}{N_{2}\left(P+N_{1}\right)}\right).
\end{IEEEeqnarray}
It is easy to see that for $k^{2}=\frac{1}{\sigma^{2}_{t}}\left[\frac{PN_{2}}{P+N_{2}}2^{2I_{\epsilon}}-\frac{PN_{1}}{P+N_{1}}\right]$, always we have $I(U;V)-I(U;V^{'})<I(U;V^{'},Y)-I(U;Z)$. The intended decoder therefore can estimate the transmitted signal $v^{n}$ from $u^{n}$, $v^{'n}$, and $y^{n}$ as follows:
\begin{IEEEeqnarray}{rl}
\hat{v}=\left(
              \begin{array}{ccc}
                \lambda_{1} & \lambda_{2} & \lambda_{3} \\
              \end{array}
            \right)\left(
                     \begin{array}{c}
                       u^{n} \\
                      v^{'n} \\
                       y^{n} \\
                     \end{array}
                   \right),
\end{IEEEeqnarray}
where $\left(
              \begin{array}{ccc}
                \lambda_{1} & \lambda_{2} & \lambda_{3} \\
              \end{array}
            \right)$ are the coefficients of the linear MMSE estimate which minimize $D=E\left[\|V-\hat{V}\|^{2}\right]$. After some math, the optimum choices for the coefficients the the related distortion given as
\begin{IEEEeqnarray}{rl}
\lambda_{1}&=\frac{k\sigma_{t}^{2}}{k^{2}\sigma_{t}^{2}+\frac{PN_{1}}{P+N_{1}}}\\ \nonumber
\lambda_{2}&=\frac{\frac{PN_{1}}{P+N_{1}}}{k^{2}\sigma_{t}^{2}+\frac{PN_{1}}{P+N_{1}}}\\ \nonumber
\lambda_{3}&=\frac{-Pk\sigma_{t}^{2}}{k^{2}\sigma_{t}^{2}\left(P+N_{1}\right)+PN_{1}}
\end{IEEEeqnarray}

\begin{IEEEeqnarray}{rl}
D&=\frac{\sigma_{t}^{2}}{\frac{P+N_{1}}{P+N_{2}}\frac{N_{2}}{N_{1}}2^{2I_{\epsilon}}}\\ \nonumber &= D^{*}
\end{IEEEeqnarray}

\subsection{Superimposed Secure Digital and Secure Hybrid Wyner-Ziv Coding Scheme}

In this scheme, the transmitted signal is a superposition of two signals $X_{1}$ and $X_{2}$, which are the outputs of a digital Wyner-Ziv encoder and a hybrid digital-analog encoder, respectively. Fig.3. illustrates the streams of the transmitter for this scheme.
\begin{figure}
\leftline{\includegraphics[scale=.2]{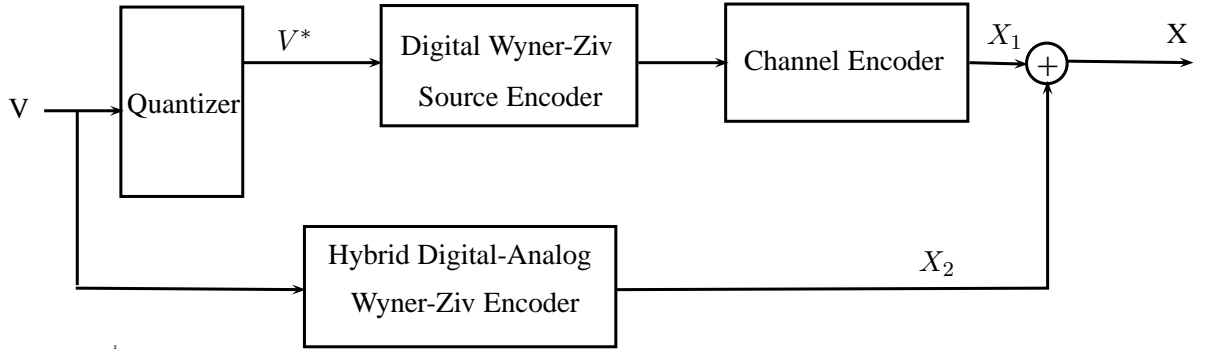}} \caption{Block Diagram of the Encoder of the Superimposed Digital and Hybrid Digital-Analog Wyner-Ziv Scheme}
\end{figure}

The first stream uses a digital rate $R$ Wyner-Ziv code to quantize the source signal assuming that the side information $V^{'}$ is known at the receiver. The discrete quantized index is then encoded using an optimal wiretap channel code to produce the codeword $X_{1}$. We allocate the following power to this stream
\begin{IEEEeqnarray}{rl}
P_{WZ}=\frac{\left(P+N_{1}\right)\left(P+N_{2}\right)\left(1-2^{-2R}\right)}{P+N_{2}-\left(P+N_{1}\right)2^{-2R}}.
\end{IEEEeqnarray}

The second stream uses the previous proposed hybrid digital-analog Wyner-Ziv scheme and produces the output signal $X_{2}$. The auxiliary random variable of this scheme is given by
\begin{IEEEeqnarray}{rl}\label{newu}
U=k_{1}V+X_{2},
\end{IEEEeqnarray}
where $X_{2}$ and $V$ are independent and $X_{2}\sim\mathcal{N}(0,P_{HWZ})$, where the allocated power to this stream is given by
\begin{IEEEeqnarray}{rl} \label{phwz}
P_{HWZ}=\frac{N_{2}\left(P+N_{1}\right)2^{-2R}-N_{1}\left(P+N_{2}\right)}{P+N_{2}-\left(P+N_{1}\right)2^{-2R}}.
\end{IEEEeqnarray}
We also chose $k_{1}^{2}$ as follows
\begin{IEEEeqnarray}{rl}\label{newk}
k_{1}^{2}=\frac{1}{\sigma^{2}_{t}}\left[\frac{P_{HWZ}N_{2}}{P_{HWZ}+N_{2}}2^{2I_{\epsilon}}-\frac{P_{HWZ}N_{1}}{P_{HWZ}+N_{1}}\right],
\end{IEEEeqnarray}
where $\sigma_{\widetilde{t}}^{2}=\sigma_{t}^{2}2^{-2R}$.

The two streams of $X_{1}$ and $X_{2}$ are superimposed at the transmitter through the channel. The received signal at the intended receiver is is given by $Y=X_{1}+X_{2}+W$. The intended decoder first decodes $X_{1}$ by assuming that $X_{2}+W$ is an independent noise and extracts the Wyner-Ziv encoded bits (index). Next, along with the side information $V^{'}$, the intended decoder maked an estimate of the source. Let us denote this estimation by $\widetilde{V}$. The source signal and the estimated signal of $\widetilde{V}$ are related as follows:
\begin{IEEEeqnarray}{rl}\label{newside}
V=\widetilde{V}+\widetilde{T},
\end{IEEEeqnarray}
where $\sigma^{2}_{\widetilde{t}}$ can be calculated as follows. Let us denote the quantized signal by $V^{*}$. Thus,
\begin{IEEEeqnarray}{rl}
V=V^{*}+ E,
\end{IEEEeqnarray}
where $E$ is the quantization error and $E\sim\mathcal{N}(0,\sigma^{2}_{v}2^{-2R})$ is independent of $V^{*}$. The intended receiver knows the signals $V^{*}$ and $V^{'}$, where
\begin{IEEEeqnarray}{rl}
V=V^{'}+T.
\end{IEEEeqnarray}
Let $\mbox{\boldmath$\Lambda$}$ be the covariance matrix of $(V^{*},V^{'})^{T}$ and $\mbox{\boldmath$\Gamma$}$ be the correlation between $V$ and $(V^{*},V^{'})^{T}$. Then $\mbox{\boldmath$\Lambda$}$ and $\mbox{\boldmath$\Gamma$}$ are given by
\begin{IEEEeqnarray}{rl}
\mbox{\boldmath$\Lambda$}=\left(
                            \begin{array}{cc}
                              \sigma^{2}_{v^{*}} & \frac{\sigma^{2}_{v^{*}}\sigma^{2}_{v^{'}}}{\sigma^{2}_{v}} \\
                              \frac{\sigma^{2}_{v^{*}}\sigma^{2}_{v^{'}}}{\sigma^{2}_{v}}&  \sigma^{2}_{v^{'}} \\
                            \end{array}
                          \right)~~\hbox{and}~~
                       \mbox{\boldmath$\Gamma$}=\left( \sigma^{2}_{v^{*}}~~~ \sigma^{2}_{v^{'}}\right)
\end{IEEEeqnarray}
Therefore,
\begin{IEEEeqnarray}{rl}\label{ttilde}
\sigma^{2}_{\widetilde{t}}&=\sigma^{2}_{v}-\mbox{\boldmath$\Gamma$}^{T}\mbox{\boldmath$\Lambda$}^{-1}\mbox{\boldmath$\Gamma$}\\ \nonumber &=\sigma^{2}_{t}2^{-2R}.
\end{IEEEeqnarray}

When the digital part is first decoded and canceled from the intended receive signal, we have an equivalent channel for the hybrid digital-analog Wyner-Ziv scheme with power constraint of $P_{HWZ}$. The intended receiver finally makes an estimate of $V$ using hybrid digital-analog Wyner-Ziv decoder from the new side information $\widetilde{V}$. The observe equivalent channel is $Y-X_{1}$. According to the previous subsection, the achievable distortion is given as follows:
\begin{IEEEeqnarray}{rl}
D&=\frac{\sigma^{2}_{\widetilde{t}}}{\frac{P_{HWZ}+N_{1}}{P_{HWZ}+N_{2}}\frac{N_{2}}{N_{1}}2^{2I_{\epsilon}}}\\ \nonumber
&\stackrel{(a)}{=}\frac{\sigma^{2}_{t}2^{-2R}}{\frac{(P+N_{1})2^{-2R}}{P+N_{2}}\frac{N_{2}}{N_{1}}2^{2I_{\epsilon}}}\\ \nonumber
&=\frac{\sigma^{2}_{t}}{\frac{P+N_{1}}{P+N_{2}}\frac{N_{2}}{N_{1}}2^{2I_{\epsilon}}}\\ \nonumber
&=D^{*},
\end{IEEEeqnarray}
where $(a)$ follows from (\ref{phwz}). The optimal distortion $D^{*}$ can be achieved for any rate $R$ which
$0\leq R\leq C$, where $C=\frac{1}{2}\log\left(\frac{P+N_{1}}{P+N_{2}}\frac{N_{2}}{N_{1}}\right)$ is the capacity of the wiretap channel. Therefore, there are infinitely many schemes which achieve the optimal distortion. In the special case that $P_{HWZ}=P$ (and therefore $R=0$) this scheme converts to the scheme of pervious subsection. Similarly, when
$P_{HWZ}=0$ (or equivalently $R=C$), this scheme is equivalent to the separated based scheme.

\section{SNR Mismatch Analysis}

In this section, we evaluate the performance of our propose secure hybrid digital-analog Wyner-Ziv schemes for the case of $SNR$ mismatch where we design the scheme to be optimal for a designed $SNR_{1}$ such that $SNR_{2}<SNR_{1}<SNR_{1a}$, but the actual $SNR$ is $SNR_{1a}$. It is well known that separation based scheme suffers from a pronounced threshold effect; when the actual $SNR$ is worse than the designed $SNR$, the index cannot be decoded and when the actual $SNR$ is better than the designed $SNR$, the distortion is limited by quantization and therefore, the distortion does not improve. We show that however our proposed secure joint source-channel coding schemes offer better performance in this situation.

Let us first consider our first proposed secure joint source-channel coding scheme. The intended receiver can decode $u^{n}$ when the $SNR_{1a}$ is better than the designed $SNR_{1}$ and make an estimate of the source from the observations at the receiver. The signals are given as follows:

\begin{IEEEeqnarray}{rl}
U&=X+kV\\ \nonumber
V&=V^{'}+T\\ \nonumber
Y&=X+W_{a}\\ \nonumber
Z&=X+W^{'},
\end{IEEEeqnarray}
where  $k=\sqrt{\frac{1}{\sigma^{2}_{t}}\left[\frac{PN_{2}}{P+N_{2}}2^{2I_{\epsilon}}-\frac{PN_{1}}{P+N_{1}}\right]}$, $W\sim\mathcal{N}(0,N_{1a})$. The intended receiver uses an optimal linear MMSE to estimate the transmitted signal $V$ from the observations of $[U,V^{'},Y]$. Note that this receiver knows the exact value of $N_{1a}$, but the transmitter chooses the parameter $k$ based on the designed $N_{1}$. Let $\mbox{\boldmath$\Lambda$}$ be the covariance matrix of $(U,V^{'},Y)^{T}$ and $\mbox{\boldmath$\Gamma$}$ be the correlation between $V$ and $(U,V^{'},Y)^{T}$. Thus,

\begin{IEEEeqnarray}{rl}
\mbox{\boldmath$\Lambda$}=\left(
                            \begin{array}{ccc}
                              P+k^{2}\sigma_{v}^{2} & k(\sigma_{v}^{2}-\sigma_{t}^{2}) & P\\
                              k(\sigma_{v}^{2}-\sigma_{t}^{2}) & \sigma_{v}^{2}-\sigma_{t}^{2} & 0 \\
                              P & 0 & P+N_{1a} \\
                            \end{array}
                          \right)
\end{IEEEeqnarray}
and
\begin{IEEEeqnarray}{rl}
\mbox{\boldmath$\Gamma$}=\left(
                           \begin{array}{ccc}
                            k\sigma_{v}^{2}& \sigma_{v}^{2}-\sigma_{t}^{2} & 0 \\
                           \end{array}
                         \right)^{T}.
\end{IEEEeqnarray}
The coefficients of the linear MMSE estimate are given by $\mbox{\boldmath$\Lambda$}^{-1}\mbox{\boldmath$\Gamma$}$. After some math, the actual distortion is then given by
\begin{IEEEeqnarray}{rl}
D_{a}&=\sigma_{v}^{2}-\mbox{\boldmath$\Gamma$}^{T}\mbox{\boldmath$\Lambda$}^{-1}\mbox{\boldmath$\Gamma$}\\ \nonumber &=\sigma_{v}^{2}-\frac{k^{2}\sigma_{v}^{2}\sigma_{t}^{2}\left(P+N_{1a}\right)+\left(\sigma_{v}^{2}-\sigma_{t}^{2}\right)PN_{1a}}{k^{2}\sigma_{t}^{2}\left(P+N_{1a}\right)+PN_{1a}}\\ \nonumber &=\frac{\sigma_{t}^{2}PN_{1a}}{k^{2}\sigma_{t}^{2}\left(P+N_{1a}\right)+PN_{1a}}\\ \nonumber &\stackrel{(a)}{=}\frac{\sigma_{t}^{2}N_{1a}}{\left[\frac{N_{2}}{P+N_{2}}2^{2I_{\epsilon}}-\frac{N_{1}}{P+N_{1}}\right]\left[P+N_{1a}\right]+N_{1a}},
\end{IEEEeqnarray}
where $(a)$ follows by substituting $k$. As we can see, this scheme has the distortion exponent $\zeta=1$.

Now let us consider the superimposed secure digital and secure hybrid Wyner-Ziv coding scheme. The intended receiver can decode the digital part ($X_{1}$) when the actual $SNR_{1a}$ is better than the design $SNR_{1}$. When the digital part is decoded and canceled from the intended receiver signal of $Y=X_{1}+X_{2}+W_{a}$, we have an equivalent hybrid digital-analog Wyner-Ziv coding scheme with the side information of $\widetilde{V}$ (see equation (\ref{newside})) available at the intended receiver and the power constraint of $P_{HWZ}$ (see equation (\ref{phwz})). The achievable distortion at the intended receiver is therefore given by

\begin{IEEEeqnarray}{rl}
D_{a}=\sigma_{v}^{2}-\mbox{\boldmath$\Gamma$}^{T}\mbox{\boldmath$\Lambda$}^{-1}\mbox{\boldmath$\Gamma$},
\end{IEEEeqnarray}
where here $\mbox{\boldmath$\Lambda$}$ is the covariance matrix of $(U,\widetilde{V},Y-X_{1})^{T}$ and $\mbox{\boldmath$\Gamma$}$ is the correlation between $V$ and $(U,\widetilde{V},Y-X_{1})^{T}$, where $U$ and $k_{1}$ are given in (\ref{newu}) and (\ref{newk}), respectively. Thus,
\begin{IEEEeqnarray}{rl}
\mbox{\boldmath$\Lambda$}=\left(
                            \begin{array}{ccc}
                              P_{HWZ}+k_{1}^{2}\sigma_{v}^{2} & k_{1}(\sigma_{v}^{2}-\sigma_{\widetilde{t}}^{2}) & P_{HWZ}\\
                              k_{1}(\sigma_{v}^{2}-\sigma_{\widetilde{t}}^{2}) & \sigma_{v}^{2}-\sigma_{\widetilde{t}}^{2} & 0 \\
                              P_{HWZ} & 0 & P_{HWZ}+N_{1a} \\
                            \end{array}
                          \right)
\end{IEEEeqnarray}
and
\begin{IEEEeqnarray}{rl}
\mbox{\boldmath$\Gamma$}=\left(
                           \begin{array}{ccc}
                            k_{1}\sigma_{v}^{2}& \sigma_{v}^{2}-\sigma_{\widetilde{t}}^{2} & 0 \\
                           \end{array}
                         \right)^{T}.
\end{IEEEeqnarray}
The coefficients of the linear MMSE estimate are given by $\mbox{\boldmath$\Lambda$}^{-1}\mbox{\boldmath$\Gamma$}$. After some math, the actual distortion is then given by
\begin{IEEEeqnarray}{rl}
D_{a}&=\sigma_{v}^{2}-\mbox{\boldmath$\Gamma$}^{T}\mbox{\boldmath$\Lambda$}^{-1}\mbox{\boldmath$\Gamma$}\\ \nonumber &=\sigma_{v}^{2}-\frac{k_{1}^{2}\sigma_{v}^{2}\sigma_{\widetilde{t}}^{2}\left(P_{HWZ}+N_{1a}\right)+\left(\sigma_{v}^{2}-\sigma_{\widetilde{t}}^{2}\right)P_{HWZ}N_{1a}}{k_{1}^{2}\sigma_{\widetilde{t}}^{2}\left(P_{HWZ}+N_{1a}\right)+P_{HWZ}N_{1a}}\\ \nonumber &=\frac{\sigma_{\widetilde{t}}^{2}P_{HWZ}N_{1a}}{k_{1}^{2}\sigma_{\widetilde{t}}^{2}\left(P_{HWZ}+N_{1a}\right)+P_{HWZ}N_{1a}}\\ \nonumber &\stackrel{(a)}{=}\frac{\sigma_{\widetilde{t}}^{2}N_{1a}}{\left[\frac{N_{2}}{P_{HWZ}+N_{2}}2^{2I_{\epsilon}}-\frac{N_{1}}{P_{HWZ}+N_{1}}\right]\left[P_{HWZ}+N_{1a}\right]+N_{1a}}\\ \nonumber &\stackrel{(b)}{=}\frac{\sigma_{t}^{2}2^{-2R}}{\left[\frac{N_{2}}{N_{1a}\left(P_{HWZ}+N_{2}\right)}2^{2I_{\epsilon}}-\frac{N_{1}}{N_{1a}\left(P_{HWZ}+N_{1}\right)}\right]\left[P_{HWZ}+N_{1a}\right]+1}\\ \nonumber &\stackrel{(c)}{=}\frac{\sigma_{t}^{2}2^{-2R}}{\frac{N_{2}2^{2I_{\epsilon}}\left(N_{2}-N_{1a}\right)\left(P+N_{1}\right)2^{-2R}+\left(P+N_{2}\right)\left(N_{1a}-N_{1}\right)}{N_{1a}\left(N_{2}-N_{1}\right)\left(P+N_{2}\right)}-\frac{N_{1}\left(N_{2}-N_{1a}\right)\left(P+N_{1}\right)2^{-2R}+\left(P+N_{2}\right)\left(N_{1a}-N_{1}\right)}{N_{1a}\left(N_{2}-N_{1}\right)\left(P+N_{1}\right)2^{-2R}}+1} \end{IEEEeqnarray}
where $(a)$ follows by substituting $k_{1}$, $(b)$ follows by substituting $\sigma_{\widetilde{t}}^{2}=\sigma_{t}^{2}2^{-2R}$, and $(c)$ follows by substituting $P_{HWZ}$ from (\ref{phwz}). As we can see, this scheme has the distortion exponent $\zeta=1$.

It is useful to mention that for the uncoded scheme of section III-B, the distortion is given by (\ref{uncd}) which may be written as follows:
\begin{IEEEeqnarray}{rl}
D_{u}=\frac{\sigma^{2}_{t}}{1+\frac{\sigma_{t}^{2}}{\sigma_{v}^{2}}\frac{N_{2}}{P}\left(2^{2I_{\epsilon}}-1\right)SNR_{1a}}.
\end{IEEEeqnarray}
Though the above equation shows that the distortion exponent $\zeta$ is 1, we have a considerable loss in optimality at the intended receiver, as we have not used the full power $P$ at the transmitter.

\begin{figure}
\centerline{\includegraphics[scale=.7]{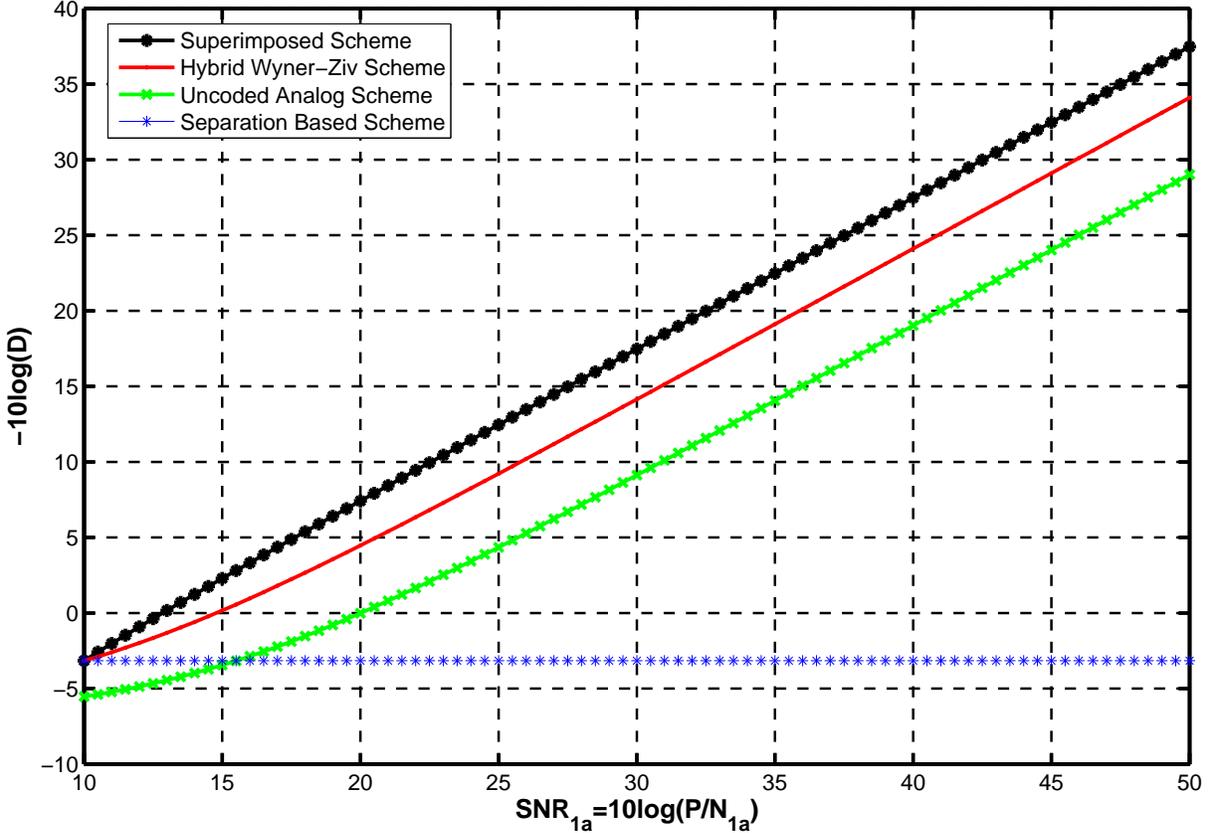}} \caption{Performance of the secure hybrid digital-analog Wyner-Ziv coding scheme compared with the separation based scheme.}
\end{figure}

Fig. 4 Compares the performance of the secure hybrid coding scheme, superimposed scheme and the analog coding scheme with the separation based scheme. In this figure, $P=10$, $SNR_{1}=10$, $SNR_{2}=7$, $10\leq SNR_{1a}\leq 50$, $I_{\epsilon}=0.2$, $R=0.15$, $\sigma^{2}_{v}=8$, and $\sigma^{2}_{t}=5$. As shown in this figure, our proposed hybrid schemes provide the optimum distortion when the transmitter has the exact value of the $SNR_{1a}$ and is more robust against the $SNR$ mismatch compared with digital Wyner-Ziv Coding scheme.

\section{Conclusions}

We considered the problem of transmitting an i.i.d Gaussian source over an i.i.d Gaussian wiretap channel with side information available at the intended receiver. We showed that Shannon's source-channel separation coding scheme is optimum in the sense of achieving the minimum distortion. We then proposed two hybrid digital-analog secure joint source-channel coding schemes which achieve the minimum distortion. Our coding schemes were based on the Wyner Ziv coding scheme and wiretap channel coding scheme when the analog source is not explicitly quantized. We analyzed our secure hybrid digital-analog schemes under the main channel SNR mismatch and showed that that our proposed schemes can give a graceful degradation of distortion with SNR under SNR mismatch, i.e., when the actual SNR is larger than the designed SNR.


\begin{thebibliography}{9}
\bibitem{1}
C. E. Shannon, ``Communication Theory of Secrecy Systems", {\em Bell
System Technical Journal}, vol. 28, pp. 656-715, October 1949.
\bibitem{2}
A. Wyner, ``The Wire-tap Channel", {\em Bell System Technical
Journal}, vol. 54, pp. 1355-1387, 1975
\bibitem{3}
I. Csisz´ar and J. K¨orner, ``Broadcast Channels with Confidential
Messages", {\em IEEE Trans. Inform. Theory}, vol. 24, no. 3, pp.
339-348, May 1978.

\bibitem{4}
A. Khisti, G. Wornell, A. Wiesel, and Y. Eldar, ``On the Gaussian MIMO Wiretap Channel", {\em in Proc. IEEE International Symposium on Information Theory} , pp.2471-2475, June 2007.
\bibitem{5}
G. Bagherikaram, A. S. Motahari and A. K. Khandani, ``The Secrecy Capacity Region of the Degraded Vector Gaussian Broadcast Channel", {\em in Proc. IEEE International Symposium on Information Theory}, pp.2772-2776, July 2009.
\bibitem{6}
P. K. Gopala, L. Lai and H. El-Gamal, ``On the Secrecy Capacity of Fading Channels", vol. 54, no. 10, pp.4687-4698, October 2008.
\bibitem{7}
Y. Liang; H.V. Poor and S. Shamai, ``Secure Communication over Fading Channels", {\em IEEE Trans. Inf. Theory}, Volume 54, Issue 6 pp: 2470-2492, 2008.
\bibitem{8}
R. Liu, H. V. Poor, ``Secrecy Capacity Region of a Multiple-Antenna Gaussian Broadcast Channel With Confidential Messages", {\em IEEE Trans. Inform. Theory, Volume}, vol. 55, Issue 3, pp.1235-1249, March 2009.


\bibitem{9}
C. Mitrpant, A. J. Han Vinck,  ``An Achievable Region for the Gaussian Wiretap Channel with Side Information ", {\em IEEE Trans. Inf. Theory}, vol. 52, no.5, pp.2181-2190, May 2006.
\bibitem{10}
Y. Chen and H. Vinck, ``Wiretap channel with side information," {\em in Proc. Int. Symp. Inform. Theory}, pp. 2607-2611, Jully 2006.

\bibitem{11}
S. K. Leung-Yan-Cheong and M. E. Hellman, ``Gaussian Wiretap Channel", {\em IEEE Trans. Inform. Theory,} vol. 24, no. 4, pp. 451-456, July 1978.

\bibitem{12}
C. E. Shannon, ``Communication in the Presence of Noise ", {\em Proc. IRE}, vol. 39, pp. 10-21, Jan. 1949.
\bibitem{13}
R. J. McAulay and D. J. Sakrison, `` A PPM/PM Hybrid Modulation System", {\em IEEE Trans. Commun. Technol.}, vol. COM-17, pp. 458-469, Aug. 1969.
\bibitem{14}
S. Vembu, S. Verdu, and Y. Steinberg, ``The Source-Channel Separation Theorem Revisited," {\em IEEE Trans. Inform. Theory,} vol. 41, pp. 44-54, Jan. 1995.
\bibitem{15}
Z. Reznic, M. Feder, and R. Zamir, ``Distortion Bounds for Broadcasting With Bandwidth Expansion," {\em IEEE Trans. Inform. Theory,} vol. 52, no. 8,  pp. 3778-3788, Aug. 2006.
\bibitem{16}
B. Chen and G. Wornell, ``Analog error-correcting codes based on chaotic dynamical systems," {\em IEEE Trans. Commun.,} vol. 46, no. 7, pp.
881-890, Jul. 1998.
\bibitem{17}
U. Mittal and N. Phamdo, ``Joint source-channel codes for broadcasting
and robust communication," {\em IEEE Trans. Inf. Theory,} vol. 48, no. 5, pp.
1082-1103, May 2002.
\bibitem{18}
Z. Reznic, R. Zamir, and M. Feder, ``Joint source-channel coding of a Gaussian-mixture source over the Gaussian broadcast channel," {\em IEEE
Trans. Inf. Theory,} vol. 48, no. 3, pp. 776-781, Mar. 2002.
\bibitem{19}
N. Merhav, S. Shamai, ``On joint source-channel coding for the Wyner-Ziv source and the Gel'fand-Pinsker channel,"
{\em IEEE Trans. Inform. Theory,} vol. 49, no. 11, pp. 2844-2855, Nov. 2003.
\bibitem{19_1}
J. Lim and D. L. Neuhoff, ``Joint and tandem source-channel coding
with complexity and delay constraints," {\em IEEE Trans. Commun.,} vol. 51,
pp. 757–766, May 2003.
\bibitem{19_2}
H. Coward, ``Joint source-channel coding: development of methods
and utilization in image communications," Ph.D thesis, Norwegian
University of Science and Engineering (NTNU), 2001.
\bibitem{19_3}
H. Coward and T. A. Ramstad, ``Robust image communication using
bandwidth reducing and expanding mappings," {\em in Proc. 34th Asilomar
Conf.,} pp. 1384-1388, Oct.–Nov. 2000.
\bibitem{19_4}
A. Fuldseth and T. A. Ramstad, ``Bandwidth compression for continuous
amplitude channels based on vector approximation to a continuous
subset of the source signal space," {\em in Proc. IEEE ICASSP,} pp. 3093–3096,  Apr. 1997.
\bibitem{19_5}
F. Hekland, G. E. Øien, and T. A. Ramstad, ``Using 2:1 Shannon mapping
for joint source-channel coding," {\em in Proc. IEEE Data Compression
Conf.,} pp. 223–232, Mar. 2005,.
\bibitem{19_6}
J. M. Lervik, A. Grovlen, and T. A. Ramstad, ``Robust digital signal
compression and modulation exploiting the advantages of analog communications,"
{\em in Proc. IEEE GLOBECOM,} pp. 1044–1048, Nov. 1995.
\bibitem{20}
M. Skoglund, N. Phamdo, and F. Alajaji, ``Design and Performance of vq-Based Hybrid Digital-Analog Joint Source Channel Codes," {\em IEEE Trans. Inform. Theory,} vol 48, no. 3, pp. 708-720, March 2002.
\bibitem{21}
M. P. Wilson, and K. Narayanan, `` Transmitting an Analog Gaussian Source Over a Gaussian Wiretap Channel Under SNR Mismatch," {\em Proc. IEEE International Conference on Telecommunications,} pp.44-47, April 2010.
\bibitem{22}
M. P. Wilson, and K. Narayanan, ``Joint Source Channel Coding with Side Information Using Hybrid Digital Analog Codes," {\em Proc. IEEE Information Theory and Application Workshop}, pp.  299-308, Feb. 2007.
\bibitem{23}
A. D. Wyner and J. Ziv, ``The rate distortion function for source coding with side information at the decoder," {\em IEEE Tran. Info. Theory,} vol. IT-22,
No. 1, pp. 1-10, Jan. 1976.
\bibitem{24}
S. Bross ,A. Lapidoth, and S. Tinguely, ``Superimposed Coded and Uncoded Transmissions of a Gaussian Source over the Gaussian Channel," {\em Proceedings of the IEEE International Symposium on Information
Theory (ISIT),} pp. 2153-2155, July 2006.
\bibitem{25}
G. Bagherikaram, K. N. Plataniotis, `` Secure Joint Source-Channel Coding With Side
Information," available at

$http://arxiv.org/PS_cache/arxiv/pdf/1008/1008.0223v1.pdf.$
\end{thebibliography}
\end{document}